\documentclass[11pt]{article}

\usepackage{amsmath}
\usepackage{amsthm}
\usepackage{amssymb}
\usepackage{comment}
\usepackage{enumerate}
\usepackage{url}

\newcommand{\ket}[1]{| #1 \rangle}

\newcommand{\ip}[2]{\langle #1|#2 \rangle}
\newcommand{\bracket}[3]{\langle #1|#2|#3 \rangle}

\newcommand{\uinvnorm}{|\kern-1pt|\kern-1pt|}

\newtheorem{thm}{Theorem}
\newtheorem*{thm*}{Theorem}

\newtheorem{lem}[thm]{Lemma}
\newtheorem*{lem*}{Lemma}

\newtheorem{fact}[thm]{Fact}

\newcommand{\be}{\begin{equation}}
\newcommand{\ee}{\end{equation}}
\newcommand{\bea}{\begin{eqnarray}}
\newcommand{\eea}{\end{eqnarray}}
\newcommand{\bes}{\begin{equation*}}
\newcommand{\ees}{\end{equation*}}
\newcommand{\beas}{\begin{eqnarray*}}
\newcommand{\eeas}{\end{eqnarray*}}

\newcommand{\Z}{\mathbb{Z}}

\title{Nonadaptive quantum query complexity}

\author{Ashley Montanaro\footnote{Department of Computer Science, University of Bristol, Woodland Road, Bristol, BS8 1UB, UK; {\tt montanar@cs.bris.ac.uk}.}}

\date{\today}

\begin{document}
	
\maketitle

\begin{abstract}
We study the power of nonadaptive quantum query algorithms, which are algorithms whose queries to the input do not depend on the result of previous queries. First, we show that any bounded-error nonadaptive quantum query algorithm that computes some total boolean function depending on $n$ variables must make $\Omega(n)$ queries to the input in total. Second, we show that, if there exists a quantum algorithm that uses $k$ nonadaptive oracle queries to learn which one of a set of $m$ boolean functions it has been given, there exists a nonadaptive classical algorithm using $O(k \log m)$ queries to solve the same problem. Thus, in the nonadaptive setting, quantum algorithms can achieve at most a very limited speed-up over classical query algorithms.
\end{abstract}


\section{Introduction}

Many of the best-known results showing that quantum computers outperform their classical counterparts are proven in the {\em query complexity} model. This model studies the number of queries to the input $x$ which are required to compute some function $f(x)$. In this work, we study two broad classes of problem that fit into this model.

In the first class of problems, {\em computational} problems, one wishes to compute some boolean function $f(x_1,\dots,x_n)$ using a small number of queries to the bits of the input $x \in \{0,1\}^n$. The query complexity of $f$ is the minimum number of queries required for any algorithm to compute $f$, with some requirement on the success probability. The deterministic query complexity of $f$, $D(f)$, is the minimum number of queries that a deterministic classical algorithm requires to compute $f$ with certainty. $D(f)$ is also known as the decision tree complexity of $f$. Similarly, the randomised query complexity $R_2(f)$ is the minimum number of queries required for a randomised classical algorithm to compute $f$ with success probability at least $2/3$. The choice of $2/3$ is arbitrary; any constant strictly between $1/2$ and 1 would give the same complexity, up to constant factors.

There is a natural generalisation of the query complexity model to quantum computation, which gives rise to the exact and bounded-error quantum query complexities $Q_E(f)$, $Q_2(f)$ (respectively). In this generalisation, the quantum algorithm is given access to the input $x$ through a unitary oracle operator $O_x$. Many of the best-known quantum speed-ups can be understood in the query complexity model. Indeed, it is known that, for certain partial functions $f$ (i.e.\ functions where there is a promise on the input), $Q_2(f)$ may be exponentially smaller than $R_2(f)$ \cite{simon97}. However, if $f$ is a total function, $D(f) = O(Q_2(f)^6)$~\cite{beals01}. See \cite{buhrman02,hoyer05} for good reviews of quantum and classical query complexity.

In the second class of problems, {\em learning} problems, one is given as an oracle an unknown function $f^?(x_1,\dots,x_n)$, which is picked from a known set $\mathcal{C}$ of $m$ boolean functions $f: \{0,1\}^n \rightarrow \{0,1\}$. These functions can be identified with $n$-bit strings or subsets of $[n]$, the integers between 1 and $n$. The goal is to determine which of the functions in $\mathcal{C}$ the oracle $f^?$ is, with some requirement on the success probability, using the minimum number of queries to $f^?$. Note that the success probability required should be strictly greater than $1/2$ for this model to make sense.

Borrowing terminology from the machine learning literature, each function in $\mathcal{C}$ is known as a {\em concept}, and $\mathcal{C}$ is known as a concept class \cite{servedio01}. We say that an algorithm that can identify any $f \in \mathcal{C}$ with worst-case success probability $p$ {\em learns} $\mathcal{C}$ with success probability $p$. This problem is known classically as exact learning from membership queries \cite{angluin88,servedio01}, and also in the literature on quantum computation as the oracle identification problem \cite{ambainis04a}. Many interesting results in quantum algorithmics fit into this framework, a straightforward example being Grover's quantum search algorithm \cite{grover97}. It has been shown by Servedio and Gortler that the speed-up that may be obtained by quantum query algorithms in this model is at most polynomial \cite{servedio01}.


\subsection{Nonadaptive query algorithms}
\label{sec:nonadaptintro}

This paper considers query algorithms of a highly restrictive form, where oracle queries are not allowed to depend on previous queries. In other words, the queries must all be made at the start of the algorithm. We call such algorithms {\em nonadaptive}, but one could also call them {\em parallel}, in contrast to the usual serial model of query complexity, where one query follows another. It is easy to see that, classically, a deterministic nonadaptive algorithm that computes a function $f:\{0,1\}^n \rightarrow \{0,1\}$ which depends on all $n$ input variables must query all $n$ variables $(x_1,\dots,x_n)$. Indeed, for any $1 \le i \le n$, consider an input $x$ for which $f(x) = 0$, but $f(x \oplus e^i)=1$, where $e^i$ is the bit string which has a 1 at position $i$, and is 0 elsewhere. Then, if the $i$'th variable were not queried, changing the input from $x$ to $x \oplus e^i$ would change the output of the function, but the algorithm would not notice.

In the case of learning, the exact number of queries required by a nonadaptive deterministic classical algorithm to learn any concept class $\mathcal{C}$ can also be calculated. Identify each concept in $\mathcal{C}$ with an $n$-bit string, and imagine an algorithm $\mathcal{A}$ that queries some subset $S \subseteq [n]$ of the input bits. If there are two or more concepts in $\mathcal{C}$ that do not differ on any of the bits in $S$, then $\mathcal{A}$ cannot distinguish between these two concepts, and so cannot succeed with certainty. On the other hand, if every concept $x \in \mathcal{C}$ is unique when restricted to $S$, then $x$ can be identified exactly by $\mathcal{A}$. Thus the number of queries required is the minimum size of a subset $S \subseteq [n]$ such that every pair of concepts in $\mathcal{C}$ differs on at least one bit in $S$.

We will be concerned with the speed-up over classical query algorithms that can be achieved by nonadaptive quantum query algorithms. Interestingly, it is known that speed-ups can indeed be found in this model. In the case of computing partial functions, the speed-up can be dramatic; Simon's algorithm for the hidden subgroup problem over $\Z_2^n$, for example, is nonadaptive and gives an exponential speed-up over the best possible classical algorithm \cite{simon97}. There are also known speed-ups for computing total functions. For example, the parity of $n$ bits can be computed exactly using only $\lceil n/2 \rceil$ nonadaptive quantum queries \cite{farhi98a}. More generally, {\em any} function of $n$ bits can be computed with bounded error using only $n/2 + O(\sqrt{n})$ nonadaptive queries, by a remarkable algorithm of van Dam \cite{vandam98}. This algorithm in fact retrieves {\em all} the bits of the input $x$ successfully with constant probability, so can also be seen as an algorithm that learns the concept class consisting of all boolean functions on $n$ bits using $n/2 + O(\sqrt{n})$ nonadaptive queries.

Finally, one of the earliest results in quantum computation can be understood as a nonadaptive learning algorithm. The quantum algorithm solving the Bernstein-Vazirani parity problem \cite{bernstein97} uses one query to learn a concept class of size $2^n$, for which any classical learning algorithm requires $n$ queries, showing that there can be an asymptotic quantum-classical separation for learning problems.


\subsection{New results}

We show here that these results are essentially the best possible. First, any nonadaptive quantum query algorithm that computes a total boolean function with a constant probability of success greater than $1/2$ can only obtain a constant factor reduction in the number of queries used. In particular, if we restrict to nonadaptive query algorithms, then $Q_2(f) = \Theta(D(f))$. In the case of exact nonadaptive algorithms, we show that the factor of 2 speed-up obtained for computing parity is tight. More formally, our result is the following theorem.

\begin{thm}
\label{thm:nonadaptive}
Let $f:\{0,1\}^n \rightarrow \{0,1\}$ be a total function that depends on all $n$ variables, and let $\mathcal{A}$ be a nonadaptive quantum query algorithm that uses $k$ queries to the input to compute $f$, and succeeds with probability at least $1-\epsilon$ on every input. Then
\[ k \ge \frac{n}{2} \left(1 - 2 \sqrt{\epsilon(1-\epsilon)} \right). \]
\end{thm}

In the case of learning, we show that the speed-up obtained by the Bernstein-Vazirani algorithm \cite{bernstein97} is asymptotically tight. That is, the query complexities of quantum and classical nonadaptive learning are equivalent, up to a logarithmic term. This is formalised as the following theorem.

\begin{thm}
\label{thm:learning}
Let $\mathcal{C}$ be a concept class containing $m$ concepts, and let $\mathcal{A}$ be a nonadaptive quantum query algorithm that uses $k$ queries to the input to learn $\mathcal{C}$, and succeeds with probability at least $1-\epsilon$ on every input, for some $\epsilon < 1/2$. Then there exists a classical nonadaptive query algorithm that learns $\mathcal{C}$ with certainty using at most
\[ \frac{4k\log_2 m}{1 - 2\sqrt{\epsilon(1-\epsilon)}} \]
queries to the input.
\end{thm}


\subsection{Related work}

We note that the question of putting lower bounds on nonadaptive quantum query algorithms has been studied previously. First, Zalka has obtained a tight lower bound on the nonadaptive quantum query complexity of the unordered search problem, which is a particular learning problem \cite{zalka99}. Second, in \cite{nishimura04}, Nishimura and Yamakami give lower bounds on the nonadaptive quantum query complexity of a multiple-block variant of the ordered search problem. Finally, Koiran et al~\cite{koiran08} develop the weighted adversary argument of Ambainis~\cite{ambainis06} to obtain lower bounds that are specific to the nonadaptive setting. Unlike the situation considered here, their bounds also apply to quantum algorithms for computing partial functions.

We now turn to proving the new results: nonadaptive computation in Section \ref{sec:noncomp}, and nonadaptive learning in Section \ref{sec:nonlearn}.


\section{Nonadaptive quantum query complexity of computation}
\label{sec:noncomp}

Let $\mathcal{A}$ be a nonadaptive quantum query algorithm. We will use what is essentially the standard model of quantum query complexity \cite{hoyer05}. $\mathcal{A}$ is given access to the input $x=x_1\dots x_n$ via an oracle $O_x$ which acts on an $n+1$ dimensional space indexed by basis states $\ket{0}, \dots, \ket{n}$, and performs the operation $O_x \ket{i} = (-1)^{x_i} \ket{i}$. We define $O_x \ket{0} = \ket{0}$ for technical reasons (otherwise, $\mathcal{A}$ could not distinguish between $x$ and $\bar{x}$). Assume that $\mathcal{A}$ makes $k$ queries to $O_x$. As the queries are nonadaptive, we may assume they are made in parallel. Therefore, the existence of a nonadaptive quantum query algorithm that computes $f$ and fails with probability $\epsilon$ is equivalent to the existence of an input state $\ket{\psi}$ and a measurement specified by positive operators $\{M_0, I-M_0\}$, such that $\bracket{\psi}{O_x^{\otimes k} M_0 O_x^{\otimes k}}{\psi} \ge 1-\epsilon$ for all inputs $x$ where $f(x) = 0$, and $\bracket{\psi}{O_x^{\otimes k} M_0 O_x^{\otimes k}}{\psi} \le \epsilon$ for all inputs $x$ where $f(x) = 1$.

The intuition behind the proof of Theorem \ref{thm:nonadaptive} is much the same as that behind ``adversary'' arguments lower bounding quantum query complexity \cite{hoyer05}. As in Section \ref{sec:nonadaptintro}, let $e^j$ denote the $n$-bit string which contains a single 1, at position $j$. In order to distinguish two inputs $x$, $x\oplus e^j$ where $f(x) \neq f(x \oplus e^j)$, the algorithm must invest amplitude of $\ket{\psi}$ in components where the oracle gives information about $j$. But, unless $k$ is large, it is not possible to invest in many variables simultaneously.

We will use the following well-known fact from \cite{bernstein97}.

\begin{fact}[Bernstein and Vazirani \cite{bernstein97}]
\label{fact:distinguish}
Imagine there exists a positive operator $M \le I$ and states $\ket{\psi_1}$, $\ket{\psi_2}$ such that $\bracket{\psi_1}{M}{\psi_1} \le \epsilon$, but $\bracket{\psi_2}{M}{\psi_2} \ge 1-\epsilon$. Then $|\ip{\psi_1}{\psi_2}|^2 \le 4\epsilon(1-\epsilon)$.
\end{fact}

We now turn to the proof itself. Write the input state $\ket{\psi}$ as
\[ \ket{\psi} = \sum_{i_1,\dots,i_k} \alpha_{i_1,\dots,i_k} \ket{i_1,\dots,i_k}, \]
where, for each $m$, $0 \le i_m \le n$. It is straightforward to compute that
\[ O_x^{\otimes k} \ket{i_1,\dots,i_k} = (-1)^{x_{i_1}+\dots+x_{i_k}} \ket{i_1,\dots,i_k}.\]
As $f$ depends on all $n$ inputs, for any $j$, there exists a bit string $x^j$ such that $f(x^j) \neq f(x^j \oplus e^j)$. Then
\[ (O_{x^j} O_{x^j \oplus e^j})^{\otimes k} \ket{i_1,\dots,i_k} = (-1)^{|\{m:i_m=j\}|} \ket{i_1,\dots,i_k}; \]
in other words $(O_{x^j} O_{x^j \oplus e^j})^{\otimes k}$ negates those basis states that correspond to bit strings $i_1,\dots,i_k$ where $j$ occurs an odd number of times in the string. Therefore, we have
\beas
|\bracket{\psi}{(O_{x^j} O_{x^j \oplus e^j})^{\otimes k}}{\psi}|^2 &=& \left(\sum_{i_1,\dots,i_k} |\alpha_{i_1,\dots,i_k}|^2 (-1)^{|\{m:i_m=j\}|} \right)^2\\
&=& \left(1-2 \sum_{i_1,\dots,i_k} |\alpha_{i_1,\dots,i_k}|^2\,[|\{m:i_m=j\}| \mbox{ odd}]\right)^2\\
&=:& (1-2 W_j)^2.
\eeas
Now, by Fact \ref{fact:distinguish}, $(1-2 W_j)^2 \le 4\epsilon(1-\epsilon)$ for all $j$, so
\[ W_j \ge \frac{1}{2}\left(1 - 2 \sqrt{\epsilon(1-\epsilon)}\right). \]
On the other hand,
\beas
\sum_{j=1}^n W_j &=& \sum_{j=1}^n \sum_{i_1,\dots,i_k} |\alpha_{i_1,\dots,i_k}|^2\,[|\{m:i_m=j\}| \mbox{ odd}]\\
&=& \sum_{i_1,\dots,i_k} |\alpha_{i_1,\dots,i_k}|^2 \sum_{j=1}^n\,[|\{m:i_m=j\}| \mbox{ odd}]\\
&\le& \sum_{i_1,\dots,i_k} |\alpha_{i_1,\dots,i_k}|^2\, k = k.
\eeas
Combining these two inequalities, we have
\[ k \ge \frac{n}{2} \left(1 - 2 \sqrt{\epsilon(1-\epsilon)} \right). \]


\section{Nonadaptive quantum query complexity of learning}
\label{sec:nonlearn}

In the case of learning, we use a very similar model to the previous section. Let $\mathcal{A}$ be a nonadaptive quantum query algorithm. $\mathcal{A}$ is given access to an oracle $O_x$, which corresponds to a bit-string $x$ picked from a concept class $\mathcal{C}$. $O_x$ acts on an $n+1$ dimensional space indexed by basis states $\ket{0}, \dots, \ket{n}$, and performs the operation $O_x \ket{i} = (-1)^{x_i} \ket{i}$, with $O_x \ket{0} = \ket{0}$. Assume that $\mathcal{A}$ makes $k$ queries to $O_x$ and outputs $x$ with probability strictly greater than $1/2$ for all $x \in \mathcal{C}$.

We will prove limitations on nonadaptive quantum algorithms in this model as follows. First, we show that a nonadaptive quantum query algorithm that uses $k$ queries to learn $\mathcal{C}$ is equivalent to an algorithm using one query to learn a related concept class $\mathcal{C'}$. We then show that existence of a quantum algorithm using one query that learns $\mathcal{C'}$ with constant success probability greater than $1/2$ implies existence of a deterministic classical algorithm using $O(\log |\mathcal{C'}|)$ queries. Combining these two results gives Theorem \ref{thm:learning}.

\begin{lem}
\label{lem:simulate}
Let $\mathcal{C}$ be a concept class over $n$-bit strings, and let $\mathcal{C}^{\otimes k}$ be the concept class defined by
\[ \mathcal{C}^{\otimes k} = \{ x^{\otimes k} : x \in \mathcal{C} \}, \]
where $x^{\otimes k}$ denotes the $(n+1)^k$-bit string indexed by $0 \le i_1,\dots,i_k \le n$, with $x^{\otimes k}_{i_1,\dots,i_k} = x_{i_1} \oplus \cdots \oplus x_{i_k}$, and we define $x_0=0$. Then, if there exists a classical nonadaptive query algorithm that learns $\mathcal{C}^{\otimes k}$ with success probability $p$ and uses $q$ queries, there exists a classical nonadaptive query algorithm that learns $\mathcal{C}$ with success probability $p$ and uses at most $kq$ queries.
\end{lem}

\begin{proof}
Given access to $x$, an algorithm $\mathcal{A}$ can simulate a query of index $(x_1,\dots,x_k)$ of $x^{\otimes k}$ by using at most $k$ queries to compute $x_1 \oplus \dots \oplus x_k$. Hence, by simulating the algorithm for learning $\mathcal{C}^{\otimes k}$, $\mathcal{A}$ can learn $\mathcal{C}^{\otimes k}$ with success probability $p$ using at most $kq$ nonadaptive queries. Learning $\mathcal{C}^{\otimes k}$ suffices to learn $\mathcal{C}$, because each concept in $\mathcal{C}^{\otimes k}$ uniquely corresponds to a concept in $\mathcal{C}$ (to see this, note that the first $n$ bits of $x^{\otimes k}$ are equal to $x$).
\end{proof}

\begin{lem}
\label{lem:onequery}
Let $\mathcal{C}$ be a concept class containing $m$ concepts. Assume that $\mathcal{C}$ can be learned using one quantum query by an algorithm that fails with probability at most $\epsilon$, for some $\epsilon < 1/2$. Then there exists a classical algorithm that uses at most $(4\log_2 m) / (1 - 2\sqrt{\epsilon(1-\epsilon)})$ queries and learns $\mathcal{C}$ with certainty.
\end{lem}

\begin{proof}
Associate each concept with an $n$-bit string, for some $n$, and suppose there exists a quantum algorithm that uses one query to learn $\mathcal{C}$ and fails with probability $\epsilon < 1/2$. Then by Fact \ref{fact:distinguish} there exists an input state $\ket{\psi} = \sum_{i=0}^n \alpha_i \ket{i}$ such that, for all $x \neq y \in \mathcal{C}$,
\[ |\bracket{\psi}{O_x O_y}{\psi}|^2 \le 4\epsilon(1-\epsilon), \]
or in other words
\be
\label{eqn:almostorth}
\left( \sum_{i=0}^n |\alpha_i|^2 (-1)^{x_i + y_i} \right)^2 \le 4\epsilon(1-\epsilon). \ee
We now show that, if this constraint holds, there must exist a subset of the inputs $S \subseteq [n]$ such that every pair of concepts in $\mathcal{C}$ differs on at least one input in $S$, and $|S| = O(\log m)$. By the argument of Section \ref{sec:nonadaptintro}, this implies that there is a nonadaptive classical algorithm that learns $M$ with certainty using $O(\log m)$ queries.

We will use the probabilistic method to show the existence of $S$. For any $k$, form a subset $S$ of at most $k$ inputs between 1 and $n$ by a process of $k$ random, independent choices of input, where at each stage input $i$ is picked to add to $S$ with probability $|\alpha_i|^2$. Now consider an arbitrary pair of concepts $x \neq y$, and let $S^+$, $S^-$ be the set of inputs on which the concepts are equal and differ, respectively. By the constraint (\ref{eqn:almostorth}), we have
\[ 4\epsilon(1-\epsilon) \ge \left( \sum_{i=0}^n |\alpha_i|^2 (-1)^{x_i + y_i} \right)^2 = \left( \sum_{i\in S^+} |\alpha_i|^2 - \sum_{i\in S^-} |\alpha_i|^2 \right)^2 = \left( 1 - 2\sum_{i\in S^-} |\alpha_i|^2 \right)^2, \]
so
\[ \sum_{i\in S^-} |\alpha_i|^2 \ge \frac{1}{2} - \sqrt{\epsilon(1-\epsilon)}. \]
Therefore, at each stage of adding an input to $S$, the probability that an input in $S^-$ is added is at least $\frac{1}{2} - \sqrt{\epsilon(1-\epsilon)}$. So, after $k$ stages of doing so, the probability that none of these inputs has been added is at most $\left(\frac{1}{2} + \sqrt{\epsilon(1-\epsilon)} \right)^k$. As there are $\binom{m}{2}$ pairs of concepts $x \neq y$, by a union bound the probability that none of the pairs of concepts differs on any of the inputs in $S$ is upper bounded by
\[ \binom{m}{2}\left(\frac{1}{2} + \sqrt{\epsilon(1-\epsilon)}\right)^k \le m^2 \left(\frac{1}{2} + \sqrt{\epsilon(1-\epsilon)}\right)^k. \]
For any $k$ greater than
\[ \frac{2 \log_2 m}{\log_2 2/(1+2\sqrt{\epsilon(1-\epsilon)})} < \frac{4\log_2 m}{1 - 2\sqrt{\epsilon(1-\epsilon)}} \]
this probability is strictly less than 1, implying that there exists some choice of $S \subseteq [n]$ with $|S| \le k$ such that every pair of concepts differs on at least one of the inputs in $S$. This completes the proof.
\end{proof}

We are finally ready to prove Theorem \ref{thm:learning}, which we restate for clarity.

\begin{thm*}
Let $\mathcal{C}$ be a concept class containing $m$ concepts, and let $\mathcal{A}$ be a nonadaptive quantum query algorithm that uses $k$ queries to the input to learn $\mathcal{C}$, and succeeds with probability at least $1-\epsilon$ on every input, for some $\epsilon < 1/2$. Then there exists a classical nonadaptive query algorithm that learns $\mathcal{C}$ with certainty using at most
\[ \frac{4k\log_2 m}{1 - 2\sqrt{\epsilon(1-\epsilon)}} \]
queries to the input.
\end{thm*}

\begin{proof}
Let $O_x$ be the oracle operator corresponding to the concept $x$. Then a nonadaptive quantum algorithm $\mathcal{A}$ that learns $x$ using $k$ queries to $O_x$ is equivalent to a quantum algorithm that uses one query to $O_x^{\otimes k}$ to learn $x$. It is easy to see that this is equivalent to $\mathcal{A}$ in fact using one query to learn the concept class $\mathcal{C}^{\otimes k}$. By Lemma \ref{lem:onequery}, this implies that there exists a classical algorithm that uses at most $(4k\log_2 m) / (1 - 2\sqrt{\epsilon(1-\epsilon)})$ queries to learn $\mathcal{C}^{\otimes k}$ with certainty. Finally, by Lemma \ref{lem:simulate}, this implies in turn that there exists a classical algorithm that uses the same number of queries and learns $\mathcal{C}$ with certainty.
\end{proof}


\section*{Acknowledgements}

I would like to thank Aram Harrow and Dan Shepherd for helpful discussions and comments on a previous version. This work was supported by the EC-FP6-STREP network QICS and an EPSRC Postdoctoral Research Fellowship.



\begin{thebibliography}{10}

\bibitem{ambainis06}
A.~Ambainis.
\newblock Polynomial degree vs. quantum query complexity.
\newblock {\em J. Comput. Syst. Sci.}, 72(2):220--238, 2006.
\newblock \url{quant-ph/0305028}.

\bibitem{ambainis04a}
A.~Ambainis, K.~Iwama, A.~Kawachi, H.~Masuda, R.~Putra, and S.~Yamashita.
\newblock Quantum identification of {B}oolean oracles.
\newblock In {\em Proc. STACS 2004}, pages 93--104. Springer, 2004.
\newblock \url{quant-ph/0403056}.

\bibitem{angluin88}
D.~Angluin.
\newblock Queries and concept learning.
\newblock {\em Machine Learning}, 2(4):319--342, 1988.

\bibitem{beals01}
R.~Beals, H.~Buhrman, R.~Cleve, M.~Mosca, and R.~de~Wolf.
\newblock Quantum lower bounds by polynomials.
\newblock {\em J. ACM}, 48(4):778--797, 2001.
\newblock \url{quant-ph/9802049}.

\bibitem{bernstein97}
E.~Bernstein and U.~Vazirani.
\newblock Quantum complexity theory.
\newblock {\em SIAM J. Comput.}, 26(5):1411--1473, 1997.

\bibitem{buhrman02}
H.~Buhrman and R.~de~Wolf.
\newblock Complexity measures and decision tree complexity: a survey.
\newblock {\em Theoretical Computer Science}, 288:21--43, 2002.

\bibitem{vandam98}
W.~van Dam.
\newblock Quantum oracle interrogation: Getting all information for almost half
  the price.
\newblock In {\em Proc. 39\textsuperscript{th} Annual Symp. Foundations of
  Computer Science}, pages 362--367. IEEE, 1998.
\newblock \url{quant-ph/9805006}.

\bibitem{farhi98a}
E.~Farhi, J.~Goldstone, S.~Gutmann, and M.~Sipser.
\newblock A limit on the speed of quantum computation in determining parity.
\newblock {\em Phys. Rev. Lett.}, 81:5442--–5444, 1998.
\newblock \url{quant-ph/9802045}.

\bibitem{grover97}
L.~Grover.
\newblock Quantum mechanics helps in searching for a needle in a haystack.
\newblock {\em Phys. Rev. Lett.}, 79(2):325--328, 1997.
\newblock \url{quant-ph/9706033}.

\bibitem{hoyer05}
P.~H{\o }yer and R.~\v{S}palek.
\newblock Lower bounds on quantum query complexity.
\newblock {\em Bulletin of the European Association for Theoretical Computer
  Science}, 87:78--103, 2005.
\newblock \url{quant-ph/0509153}.

\bibitem{koiran08}
P.~Koiran, J.~Landes, N.~Portier, and P.~Yao.
\newblock Adversary lower bounds for nonadaptive quantum algorithms.
\newblock In {\em Proc.\ {WoLLIC} 2008: 15th Workshop on Logic, Language,
  Information and Computation}, pages 226--237. Springer, 2008.
\newblock \url{arXiv:0804.1440}.

\bibitem{nishimura04}
H.~Nishimura and T.~Yamakami.
\newblock An algorithmic argument for nonadaptive query complexity lower bounds
  on advised quantum computation.
\newblock In {\em Proc.\ 29th {I}nternational {S}ymposium on {M}athematical
  {F}oundations of {C}omputer {S}cience}, pages 827--838. Springer, 2004.
\newblock \url{quant-ph/0312003}.

\bibitem{servedio01}
R.~Servedio and S.~Gortler.
\newblock Quantum versus classical learnability.
\newblock In {\em Proc. 16\textsuperscript{th} Annual IEEE Conf. Computational
  Complexity}, pages 138--148, 2001.
\newblock \url{quant-ph/0007036}.

\bibitem{simon97}
D.~R. Simon.
\newblock On the power of quantum computation.
\newblock {\em SIAM J. Comput.}, 26:1474--1483, 1997.

\bibitem{zalka99}
C.~Zalka.
\newblock Grover's quantum searching algorithm is optimal.
\newblock {\em Phys. Rev. A.}, 60(4):2746--2751, 1999.
\newblock \url{quant-ph/9711070}.

\end{thebibliography}

	
\end{document}